\documentclass{article}
\usepackage{graphicx,amsmath,amssymb,amsthm} 
\newtheorem{theorem}{Theorem}

\newtheorem{lemma}[theorem]{Lemma}

\newtheorem{remark}{Remark}
\newtheorem{example}{Example}
\newtheorem{definition}{Definition}
\title{Arithmetic Supports of Lax Difference Hierarchies}
\author{Sylvain Carpentier \thanks{SC is supported by BK21 Seoul National University.}}

\begin{document}

\maketitle

\begin{abstract}
    We classify monic finite-band scalar difference operators with independent coefficients admitting infinitely many support-preserving flows. We prove that such operators are completely characterized by an arithmetic condition on their support: the exponents must form an arithmetic progression. Conversely, every arithmetic support gives rise to an infinite hierarchy of local Lax flows. As a consequence, finite-band scalar Lax hierarchies with independent coefficients are classified by three integers $(N,p,m)$, corresponding respectively to the leading order, the common difference of the support, and the number of generators. This framework recovers several classical systems, including the Toda, Volterra, Narita--Itoh--Bogoyavlensky, and Blaszak-Marciniak lattices, while simultaneously producing infinitely many additional examples. In particular, the support $(-1,1,m)$ yields a scalar difference Lax representation of the Beffa-Wang hierarchy, and its Belov-Chaltikian reduction in the case $m=2$. 
\end{abstract}
\section{Introduction}
The theory of integrable differential-difference equations has produced a rich variety of examples, including the Toda lattice , the Volterra lattice, the Narita–Itoh–
Bogoyavlensky hierarchy, the Blaszak-Marciniak hierarchy, and many of their generalizations. A common feature of these systems is the existence of a scalar difference Lax representation of the form

\begin{equation}\label{eq}
L={\mathcal{S}}^N+\sum_{i=1}^{M}u_i{\mathcal{S}}^{N_i},
\end{equation}
where $\mathcal{S}$ denotes the shift operator and $N>N_1>\cdots>N_M$. A systematic algebraic study of difference Lax operators and their associated hierarchies was developed by Kupershmidt \cite{K85}.
\\
\indent A central problem in the theory of integrable differential-difference equations is that of classification. Symmetry-based approaches have led to important classification results for classes of scalar equations \cite{Y06}. More generally, integrability is often accompanied by strong algebraic constraints arising from symmetry and reduction \cite{M80}.
In this paper we investigate a complementary problem. Rather than classifying integrable differential-difference equations with a small number of dependent variables, we ask \textit{which finite-band scalar difference operators admit infinitely many   Lax flows.} More precisely, we ask for which $L$ there exist infinitely many linearly independent auxiliary operators $A$ such that the commutator $[A,L]$ preserves the support of $L$.\\
\indent
At first sight, there seems to be no reason for the support of $L$ to satisfy any particular arithmetic structure. One might expect sporadic supports to admit isolated families of local symmetries. Our main result shows that this is not the case:  the existence of infinitely many linearly independent auxiliary operators $A$ such that $supp([A,L]) \subset supp(L)$ forces the support of $L$ to be arithmetic. Moreover, the standard hierarchy generated by powers of $L$ emerges automatically from this weak assumption.

\begin{theorem} 
Suppose that \begin{enumerate}
    \item $N>0$ and $N_M<0$ or $N<0$
    \item the coefficients $u_i$ are free generators of a commutative algebra $\mathcal{A}$ of difference polynomials
    \item there exists an infinite sequence of linearly independent operators $A_n \in \mathcal{A}[\mathcal{S},\mathcal{S}^{-1}]$ such that

\[
supp([A_n,L]) \subset supp(L).
\]
\end{enumerate}
Then the support of $L$ is arithmetic 
\[ supp(L)=\{N,N-p,N-2p,..., N-mp \}. \]
Moreover, if $q=p/gcd(N,p)$, the hierarchy associated to $L$ is 
\[ \partial_k(L)=[L^{\epsilon qk}_+,L], \, k \geq 1\]
where $\epsilon$ denotes the sign of $N$.
\end{theorem}
Thus arithmetic supports are not merely a property of known examples, they completely characterize finite-band scalar difference Lax operators with independent coefficients by three integers: the degree $N$, the size of the gaps $p$, and the number of dependent variables $m$.  
While Theorem $1$ recovers several classical examples, many parameter values do not appear to have been previously studied. The following table summarizes how known integrable systems fit in our framework.
\\
\\
\\
\begin{tabular}{l|l|l}
     Parameters &  Support & Known hierarchy
     \\
     & &
     \\
     \hline
   & & \\

   $(1,2,1)$  & $\{1,-1\} $ & Volterra 
   \\
   $(1,p>2,1)$  & $\{1,1-p\} $ & Narita-Itoh-Bogoyavlensky 
   \\
   $(1,1,2)$  & $\{1,0,-1\} $ & Toda 
   \\
   $(2,1,3)$  & $\{2,1,0,-1\} $ & Blaszak-Marciniak 
   \\
   $(N,p,m)$,   & $\{N,N-p,...,N-mp\} $ & arithmetic support 
   \\
\end{tabular}
\vspace{5 mm}
\\
The Narita-Itoh-Bogoyavlensky hierarchy was introduced in the works of Narita \cite{N82}, Itoh \cite{I87} and Bogoyavlensky \cite{B88}, while the Blaszak-Marciniak lattice is described in \cite{BM94}.
\\
\indent 
After recalling the calculus of difference operators in Section $2$, we prove Theorem $1$ in Section $3$ when $N>0$ and $N_M<0$. We show that the local Lax symmetries must be a combination of powers of $L$ and use the asymptotic of the support of these powers. Besides recovering the Toda, Volterra, and Narita-Itoh-Bogoyavlensky hierarchies, our classification produces infinitely many new examples. In particular, the support $(1,p>1,2)$ gives rise to a two-field hierarchy whose first flow is 
\[
\begin{cases}
{u}_t=(1-\mathcal{S}^{-p})(v)+u(\mathcal{S}^{p-1}+...+\mathcal{S}-\mathcal{S}^{-1}-...-\mathcal{S}^{1-p})(u),\\[0.4em]
v_t= v(\mathcal{S}^{2p-1}+...+\mathcal{S}^p-1...-\mathcal{S}^{1-p})(u).
\end{cases}
\]
Thus the $(1,p,2)$ hierarchy provides a natural generalization of the Toda lattice ($p=1$) and the Narita-Itoh-Bogoyavlensky hierarchy ($v=0$).
\\
\vspace{1 mm}
\indent In Section $4$, we prove Theorem $1$ when $N<0$. In that case, we reduce the problem to a polynomial model and establish a nonvanishing result for the residue of negative powers of sparse polynomials. This $N<0$ case also produces infinitely many new integrable hierarchies. In particular, we show that the first flow of the hierarchy with support $(-1,1,m)$ coincides with the Belov-Chaltikian equation for $m=2$ and the Beffa-Wang equation for $m>2$. This appears to provide a new scalar difference Lax representation of these hierarchies. To match the notations of \cite{DSKW18}, we need to shift the entries of the coefficients of $L$ 
\[ L=\mathcal{S}^{-1}+v_2 \mathcal{S}^{-2}+\mathcal{S}^{-1}(v_3) \mathcal{S}^{-3}+...+\mathcal{S}^{3-m}(v_{m-1}) \mathcal{S}^{1-m}+\mathcal{S}^{-m-2}(v_{1}) \mathcal{S}^{-m}.\]
The present work raises several natural questions. It would be interesting to investigate whether analogous rigidity phenomena persist in more general classes of Lax operators such as matrix and rational difference operators. Moreover, we intend to understand the Hamiltonian structures associated to these hierarchies as well as their reductions.

\newpage 

 \section{Difference and Pseudodifference Operators} 
 In this section, we recall basic definitions about difference operators on a difference algebra $\mathcal{A}$. We assume that $\mathcal{A}$ is a commutative algebra over $\mathbb{C}$, together with an automorphism $\mathcal{S}$.
 \begin{definition}
     The algebra of difference (resp. pseudodifference) operators on $\mathcal{A}$ is the vector space of Laurent polynomials (resp. series) $\mathcal{A}[\mathcal{S},\mathcal{S}^{-1}]$ (resp. $\mathcal{A}[\mathcal{S}]((\mathcal{S}^{-1}))$). In both algebras, the multiplication is defined by
     \[
        a \mathcal{S}^n b \mathcal{S}^m =a\mathcal{S}^n(b) \mathcal{S}^{n+m},
        \]
        for all $a, b \in \mathcal{A}$ and for all $n,m \in \mathbb{Z}$. A pseudodifference is monic if its leading term is of the form $\mathcal{S}^n$ for $n \in \mathbb{Z}$.
 \end{definition}
 \begin{lemma}
     Every monic difference operator admits a unique inverse in the algebra of pseudodifference operators.
 \end{lemma}
 \begin{proof}
     After factoring out the leading term, the geometric series is well-defined in $\mathcal{A}[\mathcal{S}]((\mathcal{S}^{-1}))$. 
 \end{proof}
 \begin{definition}
     The residue $\text{Res } P$ of a pseudodifference operator $P$ is the coefficient of $\mathcal{S}^0=1$.
 \end{definition}
 \begin{lemma}
     For any pair of pseudodifference operators $P,Q$, 
     \[
     \text{Res }[P,Q] \in \text{Im } ( \mathcal{S}-1).
     \]
 \end{lemma}
 \begin{proof}
    It suffices to check the statement for $P=a \mathcal{S}^n$ and $Q=\mathcal{S}^{-n}$ with $a,b \in \mathcal{A}$ and $n >0$. Indeed we have 
    \[ [a\mathcal{S}^n,b \mathcal{S}^{-n}]=a \mathcal{S}^n(b)-\mathcal{S}^{-n}(a)b=(\mathcal{S}-1)(\mathcal{S}^{n-1}+...+1)(\mathcal{S}^{-n}(a)b).\]
 \end{proof}
\section{Positive-Order Lax Operators}
In this section we consider the commutative algebra $\mathcal{A}$ of difference polynomials freely generated over $\mathbb{C}$ by the coefficients of
        \[ L=\mathcal{S}^{N}+u_1 \mathcal{S}^{N_1}+...+u_M \mathcal{S}^{N_M},\]
        where $N > N_1 > ... > N_M$, $N>0$ and $N_M <0$. We define the integers
        \[ p=gcd(N-N_1,...,N-N_M), \, \, q=p/gcd(p,N), \, \, m=(N-N_M)/p.\]
        We define a grading on the space of difference operators $\mathcal{A}[\mathcal{S}, \mathcal{S}^{-1}]$ by
        \begin{equation}\label{grading}
        d(\mathcal{S})=1, \, d(\mathcal{L})=N .
        \end{equation} 
        A Lax flow is a derivation $\partial$ of $\mathcal{A}$  commuting with the shift $\mathcal{S}$ whose action on the generators $u_i$ is given by 
        \[\partial(L)=[A,L], \]
        where the auxiliary operator $A \in \mathcal{A}[\mathcal{S},\mathcal{S}^{-1}]$ is such that 
        \[ supp \,  [A,L] \subset supp \,  L.\]
        We now prove Theorem $1$ in the case $N>0$.
        \begin{theorem}
            Suppose that $L$ admits infinitely many linearly independent Lax flows. Then the support of $L$ is arithmetic. Moreover, there exists a unique hierarchy of commuting Lax flows on $\mathcal{A}$
            \begin{equation}\label{eqeqeq}
            \partial_k(L)=[(L^{qk})_+,L]\, \, k >0.
            \end{equation} 
            This hierarchy possesses infinitely many conserved densities given by 
            \[ h_k= \text{ Res } L^{qk}, \, \, k >0.\]
        \end{theorem}      
Before proving this Theorem, we state and prove a few Lemmas.
\begin{lemma}
   Suppose that $A \in \mathcal{A}[\mathcal{S}, \mathcal{S}^{-1}]$ is such that $supp([A,L]) \subset supp(L)$. Then there exists $B \in \mathcal{A}[\mathcal{S}]$ such that $[B,L]=[A,L]$.
\end{lemma}
\begin{proof}
By induction on the lowest exponent of $A$. Let us denote it by $a \mathcal{S}^k$. If $k$ is nonnegative, there is nothing to prove. If it is negative, since the support of $[A,L]$ is contained in the support of $L$, we must have $[a \mathcal{S}^k,u_M \mathcal{S}^{N_M}]=0$. Indeed, this term is the unique contribution to the lowest exponent of $[A,L]$ which lies strictly below $N_M$. In other words, we must have
\[ 
a \mathcal{S}^k(u_M)=u_M \mathcal{S}^{N_M}(a).
\]
Let $n$ be maximal such that $u_M \mathcal{S}^{N_M}(u_M)... \, \mathcal{S}^{nN_M}(u_M)$ divides $a$ and let $\Tilde{a}$ be the quotient. We have 
\[ 
\Tilde{a} \mathcal{S}^k(u_M)= \mathcal{S}^{(n+1)N_M}(u_M)\mathcal{S}^{N_M}(\Tilde{a}).
\]
If $k \neq (n+1)N_M$, then $\mathcal{S}^{(n+1)N_M}(u_M)$ must divide $\Tilde{a}$ so $n$ is not maximal. Hence $k=(n+1)N_M$ and $\Tilde{a}=\mathcal{S}^{N_M}(\Tilde{a})$. Since 
\[ 
\mathcal{A}=\mathbb{C}[\mathcal{S}^j(u_i) | 1 \leq i \leq M , j \in \mathbb{Z}], 
\]
 only constants can be invariant for $\mathcal{S}^{N_M}$ which implies that $\Tilde{a} \in \mathbb{C}$. We showed that the lowest order term of $A$ must be the same as the lowest order term of $\Tilde{a}L^{n+1}$.
Since $[A,L]=[A-\Tilde{a}L^{n+1},L]$, we complete the proof by induction on the lowest exponent of $A$.
\end{proof}

\newpage

\begin{lemma}
    Let $A \in \mathcal{A}[\mathcal{S}]$ be such that $supp([A,L]) \subset supp(L)$. Then there exists $B \in \mathbb{C}[L]$ such that $A=B_+$.
\end{lemma}
\begin{proof}
    Since $L$ is homogeneous for the grading \eqref{grading}, every homogeneous component of $A$ must satisfy the Lemma hypothesis. Therefore we may assume that $A$ itself is homogeneous. Since $L$ is monic, there exists a unique family $B_i \in \mathcal{A}[\mathcal{S}]$ such that 
    \[ 
    A= \sum_{i \leq I} B_i L^i, 
    \]
    and the exponents of $B_i$ are bounded by $N-1$ and $0$.
    We have 
     \[ 
     [A,L]= \sum_{i \leq I} [B_i,L] L^i. 
     \]
     For any $i$, the order of $[B_i,L]$ is at most $2N-1$. Hence, if the leading coefficient of $B_I$ is not a constant, the leading coefficient of $[A,L]$ will have strictly greater order than $N_1$ which leaves the support of $L$ and is a contradiction. Using the homogeneity of $A$, we may therefore assume that 
     \[ 
     A= b \mathcal{S}^n L^I+\sum_{i \leq I} C_i L^i ,
     \]
     where $b \in \mathbb{C}$, $n<N$ and no coefficient of any $C_i$ is a constant. We have 
     \[ [b \mathcal{S}^n L^I, L]=b(\mathcal{S}^n-1)(u_1)S^{n+N_1+NI}+ o(n+N_1+NI).
     \]
     Moreover, if $J$ is the largest index such that $C_J \neq 0$ and $X \mathcal{S}^k$ is the leading term of $C_J$, then 
     \[ [\sum_{i \leq I} C_i L^i,L]=(1-\mathcal{S}^N)(X)S^{k+(J+1)N}+ o(p+(J+1)N).
     \]
      We know by assumption that the sum of these two commutators should have order at most $N_1$. If $n>0$, $n+N_1+NI>N_1$ no matter $I \geq 0$ so the leading term of both commutators must cancel. However, this is impossible since $(\mathcal{S}^n-1)(u_1)$ cannot be in the image of $\mathcal{S}^N-1$. Hence $n=0$ and the leading term of $[A,L]$ is given uniquely by the second term. We need $k+(J+1)N \leq N_1 <N$, hence $J$ is negative. We have therefore proved that 
      \[
      A_+=b L^I_+.
      \]
\end{proof}
\begin{lemma}
    Suppose that there exists a sequence $(A_k)_{k \in \mathbb{Z}_{+}} \in \mathcal{A}[\mathcal{S}, \mathcal{S}^{-1}]$ such that 
    \begin{enumerate}
        \item $supp([A_k,L]) \subset supp(L)$,
        \item $[A_k,L]$ spans an infinite-dimensional subspace of $\mathcal{A}[\mathcal{S}, \mathcal{S}^{-1}]$.
    \end{enumerate}
    Then the set $ X= \{ n \in \mathbb{N} | supp([L^n_+,L]) \subset supp(L) \} $ is infinite and $[L^n_+,L]$ with $n \in X$ span an infinite-dimensional subspace of  $\mathcal{A}[\mathcal{S}, \mathcal{S}^{-1}]$.
\end{lemma}
\begin{proof}
    By Lemma $5$, we can replace the sequence $A_k \in \mathcal{A}[\mathcal{S}, \mathcal{S}^{-1}]$ by a sequence $B_k \in \mathcal{A}[\mathcal{S}]$ without changing both conditions of the Lemma. By Lemma $6$, every $B_k$ must be a linear combination of $L^n_+$ for $n \in X$. Therefore $X$ must be infinite and the set of $[L^n_+,L]$ with $n \in X$ must span an infinite-dimensional subspace of  $\mathcal{A}[\mathcal{S}, \mathcal{S}^{-1}]$ in order to meet condition $2$.
\end{proof}

\begin{lemma}
    Let $i \in \mathbb{Z}$. Then, for large enough $k >0$, the coefficient of $\mathcal{S}^{ip}$ in $L^{qk}$ is nonzero.
\end{lemma}
\begin{proof}
    First, let us check that $L^{qk}$ is a series in $\mathcal{S}^p$. This is because one can write $L=\mathcal{S}^{N}M(\mathcal{S}^p)$ and $\mathcal{S}^{Nqk}=(\mathcal{S}^p)^{Nkq/p}$. Second, note that every nonzero contribution of $L^{qk}$ comes with a positive coefficient, so proving that $F(k) \neq 0$ for large $k$ amounts to saying that the diophantine equations
    \begin{equation} \label{dio}
        \begin{split}
            -ip &= a_0^k N+a_1^k N_1+...+a_M^k N_M,  \\
     qk &= a_0^k+ ...+ a_M^k,
    \end{split}
    \end{equation}
    have a solution for large $k$ in $\mathbb{Z}_+^M$. Using the variables $d_j=(N-N_j)/p$, we can rewrite the first equation as
\[a_1^kd_1+...+a_M^kd_M=qkN/p+i. \]
For each integer $s=0,...,m-1$, since 
\[ gcd(d_1,...,d_M)=1, \]
and by Bezout, there exists 
\[ B_s=b_1^sd_1+...+b_M^sd_M=s+l_sm,\]
where $l_s$ and all $b_i^s$ are nonnegative integers. Bezout gives relative integers but after adding $l_s m$ for large enough $l_s$ we can always obtain nonnegative coefficients.  From now on, let us pick integer $k$ such that 
\[ k > \frac{p}{qN}  \, max \,  \{ s+ l_sm \, | s=0,...,m-1 \}. \]
There exists a unique $s$ and $l >0$
such that 
\[ \frac{qkN}{p}+i = B_s+lm.\]
Recall that $m=d_M$. Therefore $(b_1^s,...,b_M^s+l)$ satisfies the first equation in \eqref{dio}. As for the second equation, we have 

 \[ l= \frac{Nk}{rm}+ O(1). \]
 Therefore we have 
 \[ b_1^s+...+b_M^s+ l=\frac{qkN}{mp}+ O(1),\]
 which ensures that our tuple also satisfies the second equation in \eqref{dio} for all large $k$, since 
 \[ N_M=N-mp<0.\]
 Indeed, $ b_1^s+...+b_M^s+ l < qk$ for large $k$ so one can pick a nonnegative integer $a_0^k$ such that $(a_0^k,b_1^s,...,b_M^s+l)$ satisfies both equations in \eqref{dio}.
\end{proof}
\begin{lemma}
   Suppose that $ X= \{ n \in \mathbb{N} | supp([L^n_+,L]) \subset supp(L) \} $ is infinite and that $[L^n_+,L]$ with $n \in X$ spans an infinite-dimensional subspace of  $\mathcal{A}[\mathcal{S}, \mathcal{S}^{-1}]$. Then $supp(L)$ is arithmetic.
\end{lemma}
\begin{proof}
    If $n$ is not divisible by $q$, then the support of $L$ and the support of $[L^n_+,L]$ are orthogonal. So if such a $n$ is in $X$ we must have $[L^n_+,L]=0$. However, the set of $[L^n_+,L]$ with $n \in X$ spans an infinite-dimensional subspace of  $\mathcal{A}[\mathcal{S}, \mathcal{S}^{-1}]$. Therefore there are infinitely many integers divisible by $q$ in $X$. Our goal is to prove now that necessarily we must have 
    \[ supp(L)=\{N,N-p,...,N-mp\}.\]
    By Lemma $8$, we know that for every $k$ large enough, the coefficients of $\mathcal{S}^{-p}, ..., \mathcal{S}^{-mp}$  in $L^{qk}$ which will denote by $F_1(k),...,F_m(k)$ are nonzero elements of $\mathcal{A}$ with positive coefficients. For $j=1,...,m$, the coefficient of $S^{N-jp}$ in $[L^{qk}_{-},L]$ is 
    \[ (1-\mathcal{S}^N)(F_j(k))+\sum_{i,  d_i < j} (\mathcal{S}^{(d_i-j)p}(u_i) - u^i\mathcal{S}^{N-d_ip})(F_{j-d_i}(k)). \]
    In particular, for $j=1$, this collapses to the first term which is not zero if and only if $F_1(k)$ is nonzero. This means that $d_1=1$ and $N-p \in supp(L)$.
    \\
    \\
    \underline{Case 1: $N \geq p$}
      \, \, The difference operator 
    \[\bar{L}=\mathcal{S}^N+u_1\mathcal{S}^{N-p}+u_M\mathcal{S}^{N-mp}\] also satisfies the conditions of Lemma $8$, which means that for large enough $k$, the specialization of $F_i(k)$ to $u_2=...=u_{M-1}=0$ are still all nonzero with positive coefficients. The coefficent of $S^{N-jp}$ with $m>j >1$ in $[{\bar{L}}^{qk}_{-},\bar{L}]$ is 
    \[ (1-\mathcal{S}^N)(F_j(k))+ (\mathcal{S}^{(1-j)p}(u_1) - u_1\mathcal{S}^{N-p})(F_{j-1}(k)). \]
    Define a $\mathbb{Z}$-filtration on the difference algebra generated by $u_1$ and $u_M$ by declaring that 
    \[ d'(\mathcal{S}^n(u^j))=n,\]
    for all $n \in \mathbb{Z}$.
     The top component of our $\mathcal{S}^{N-jp}$ coefficient is the top component of 
    \[-\mathcal{S}^N(F_j(k)) - u_1\mathcal{S}^{N-p}(F_{j-1}(k)).\]
    In particular it is a sum of negative terms, hence it is nonzero for large $k$. Therefore the coefficient of $S^{N-jp}$ in $[{\bar{L}}^{qk}_{-},\bar{L}]$ is nonzero, which implies that the coefficient of $S^{N-jp}$ in $[L^{qk}_{-},L]$ is nonzero. This 
    forces $N-jp$ to belong in $supp(L)$ for all $1<j<m$ which ends the proof in that case.
    \\
    \\
    \underline{Case 2: $N <p$} 
    \\
    In that case we use a different strategy. Suppose that the support of $L$ is not arithmetic and let $(N_i,N_{i+1})$ be the first jump that is greater than $p$. We specialize our algebra to $u_{i+2}=u_{i+3}=...u_M=0$. Using Lemma $8$, we see that the coefficient of $\mathcal{S}^p$ in $L^{qk}$  is nonzero for large $k$ after performing that specialization. This means that the coefficient of $\mathcal{S}^{N-d_{i+1}p+p}$ in $[L^{qk}_+,L]$ can never be zero which contradicts the existence of that gap. Indeed if it were zero we would need to have 
    \[ F(k)\mathcal{S}^p(u_{i+1})=u_{i+1}\mathcal{S}^{N-d_{i+1}p}(F(k)).\]
    This forces $F(k)$ to lie in a finite dimensional subspace of $\mathcal{A}$. This is impossible since they are all nonzero for large $k$ and their grading \eqref{grading} is $qkN-p$.
\end{proof}

\begin{proof} (\emph{of Theorem $4$})
    By Lemmas $7$ and $9$, the support of $L$ must be arithmetic. Conversely, the evolutionary derivations \eqref{eqeqeq} are well-defined and commute by a standard argument. For each pair of integers $k,k' \geq 1$, $\partial_k(h_{k'})$ lies in the image of $\mathcal{S}-1$ by Lemma $3$ since it is the residue of a commutator.
\end{proof}
\vspace{5 mm}

\begin{example}
The Toda lattice corresponds to the choice of parameters $(N=1,p=1,m=2)$, while the Narita-Itoh-Bogoyavlensky lattice corresponds to $(N=1,p,m=1)$ for any $p >2$. Let us now consider the Lax operator 
\[ L=\mathcal{S}+\mathcal{S}^p(u)\mathcal{S}^{1-p}+v \mathcal{S}^{1-2p}, \, \, p >1\]
which provides a common generalization of the Toda lattice and the Narita-Itoh-Bogoyavlensky hierarchy.
We write the first coefficient in the form $\mathcal{S}^p(u)$ in order to obtain a more symmetric system and to follow the traditional Toda notation (\cite{F74},\cite{KMW13}).
The first nontrivial flow is obtained from the auxiliary operator  
\[ (L^p)_+=\mathcal{S}^p+(\mathcal{S}^p+\mathcal{S}^{p+1}+...+\mathcal{S}^{2p-1})(u),\]
which leads to the system of differential-difference equations
\[
\begin{cases}
{u}_t=(1-\mathcal{S}^{-p})(v)+u(\mathcal{S}^{p-1}+...+\mathcal{S}-\mathcal{S}^{-1}-...-\mathcal{S}^{1-p})(u),\\[0.4em]
v_t= v(\mathcal{S}^{2p-1}+...+\mathcal{S}^p-1...-\mathcal{S}^{1-p})(u).
\end{cases}
\]
For $p=1$ this reduces to the Toda lattice, while for $v=0$ one recovers the Narita-Itoh-Bogoyavlensky lattice. By Theorem $4$, this system belongs to an infinite hierarchy of commuting Lax flows
\[\partial_k(L)=[(L^{pk})_+,L], \, \, k \geq 1.\]
Moreover, this hierarchy possesses infinitely many conserved densities given by 
\[h_k= \text{ Res } L^{pk}, \, \, k \geq 1.\]
The Hamiltonian structure of this family appears to be unknown.
\end{example}
    
\section{Negative-Order Lax Operators}
        In this section we consider the commutative algebra $\mathcal{A}$ over $\mathbb{C}$ of difference polynomials freely generated by the coefficients of
        \[ L=\mathcal{S}^{-N}+u^1 \mathcal{S}^{-N_1}+...+u^M \mathcal{S}^{-N_M}\]
        where $0<N < N_1 < ... < N_M$. We define the integers
        \[ p=gcd(N_1-N,...,N_M-N), \, \, q=p/gcd(p,N), \, \, m=(N_M-N)/p,\]
        and a grading on the space of difference operators $\mathcal{A}[\mathcal{S}, \mathcal{S}^{-1}]$ by
        \begin{equation}\label{gradingd}
        d(\mathcal{S})=1, \, d(\mathcal{L})=-N  .
        \end{equation} 
 We now prove Theorem $1$ in the case $N<0$.
        \begin{theorem}
            Suppose that $L$ admits infinitely many linearly independent Lax flows. Then the support of $L$ is arithmetic. Moreover, there exists a unique hierarchy of commuting Lax flows on $\mathcal{A}$
            \begin{equation}\label{eqeqeqeq}
            \partial_k(L)=[(L^{-qk})_+,L]\, \, k >0.
            \end{equation} 
            This hierarchy possesses infinitely many conserved densities given by 
            \[ h_k= \text{ Res } L^{-qk}, \, \, k >0.\]
        \end{theorem}

     \begin{lemma}
   Suppose that $A \in \mathcal{A}[\mathcal{S}, \mathcal{S}^{-1}]$ is such that $supp([A,L]) \subset supp(L)$. Then there exists $B \in \mathcal{A}[\mathcal{S}]$ such that $[B,L]=[A,L]$.
\end{lemma}
\begin{proof} 
The proof is identical to that of Lemma $5$. The condition  $supp([A,L]) \subset supp(L)$ forces $A_{-}$ to be of a polynomial in $L$, which commute with $L$, so 
\[ [A,L]=[A_+,L].\]
\end{proof}
\begin{lemma}
    Let $A \in \mathcal{A}[\mathcal{S}]$ be such that $supp([A,L]) \subset supp(L)$. Then there exists $B \in \mathbb{C}[L^{-1}]$ such that $A=B_+$.
\end{lemma}
\begin{proof}
    Since $L$ is homogeneous for the grading \eqref{gradingd}, every homogeneous component of $A$ must satisfy the Lemma hypothesis. Therefore we may assume that $A$ itself is homogeneous. There exists a unique family $B_i \in \mathcal{A}[\mathcal{S}]$ of order at most $N-1$ such that 
    \[ A= \sum_{i \leq I} B_i L^{-i}. \]
    We have 
     \[ [A,L]= \sum_{i \leq I} [B_i,L] L^{-i}. \]
     For any $i$, the order of $[B_i,L]$ is at most $-1$. Hence, if the leading coefficient of $B_I$ is not a constant, the leading coefficient of $[A,L]$ will have strictly greater order than $-N_1$ which is a contradiction. We may therefore assume that 
     \[ A= b \mathcal{S}^n L^{-I}+\sum_{i \leq I} C_i L^{-i}, \]
     where no coefficient of any $C_i$ is a constant by homogeneity of $A$. We have 
     \[ [b \mathcal{S}^n L^{-I}, L]=b(\mathcal{S}^p-1)(u^1)S^{n-N_1+NI}+ o(n-N_1+NI).\]
     Moreover, if $J$ is the largest index such that $C_J \neq 0$ and $X \mathcal{S}^k$ is the leading term of $C_J$, then 
     \[ [\sum_{i \leq I} C_i L^{-i},L]=(1-\mathcal{S}^{-N})(X)S^{k+(J-1)N}+ o(k+(J-1)N).\]
      We know by assumption that the sum of these two commutators should have order at most $-N_1$. If $n>0$, $n-N_1+NI>-N_1$ no matter $I \geq 0$ so the leading term of both commutators must cancel. However, this is impossible since $(\mathcal{S}^n-1)(u_1)$ cannot be in the image of $\mathcal{S}^N-1$. Hence $n=0$ and the leading term of $[A,L]$ is given uniquely by the second term. We need $k+(J-1)N \leq -N_1 <-N$, hence $J$ is negative. We have therefore proved that 
      \[ A_+=b L^{-I}_+.\]
\end{proof}
\begin{lemma}
    Suppose that there is a sequence $(A_k)_{k >0} \in \mathcal{A}[\mathcal{S}, \mathcal{S}^{-1}]$ such that 
    \begin{enumerate}
        \item $supp([A_k,L]) \subset supp(L)$,
        \item $[A_k,L]$ spans an infinite-dimensional subspace of $\mathcal{A}[\mathcal{S}, \mathcal{S}^{-1}]$.
    \end{enumerate}
    Then the set $ X= \{ n \in \mathbb{N} | supp([L^{-n}_+,L]) \subset supp(L) \} $ is infinite and $[L^{-n}_+,L]$ with $n \in X$ span an infinite-dimensional subspace of  $\mathcal{A}[\mathcal{S}, \mathcal{S}^{-1}]$.
\end{lemma}
\begin{proof}
    By Lemma $11$, we can replace the sequence $A_k \in \mathcal{A}[\mathcal{S}, \mathcal{S}^{-1}]$ by a sequence $B_k \in \mathcal{A}[\mathcal{S}]$ without changing the second condition. By Lemma $12$, every $B_k$ must be a linear combination of $L^{-n}_+$ for $n \in X$. Therefore $[L^{-n}_+,L]$ with $n \in X$ must span an infinite-dimensional subspace of  $\mathcal{A}[\mathcal{S}, \mathcal{S}^{-1}]$ by condition $2$.
\end{proof}
\begin{lemma}
    For $k$ large enough, the coefficients $F(k)$ and $G(k)$ of $\mathcal{S}^p$ and $\mathcal{S}^{-p}$ are nonzero in $L^{-qk}$.
\end{lemma}
\begin{proof}
We prove that Lemma by establishing a weaker statement. Let $v_i$ be the generators of a commutative algebra of polynomials and let 
\[P(z)=z^{-N}+v_1z^{-N_1}+....+v_Mz^{-N_M}.\]
We will show that for large enough $k$, the coefficients of $z^p$ and $z^{-p}$ are nonzero in $P^{-qk}$. This will imply our Lemma after the specialization $\mathcal{S}^k(u_i)=v_i$ for all $k \in \mathbb{Z}$ and all $i$.
Let us write $P(z)=z^{-N}(1+Q(z))$. Then we have 
\[ z^{qkN}P^{-qk}(z)=\sum_{r \geq 0}  (-1)^r\binom{k+r-1}{r}Q^r(z).\]
Note that there cannot be any cancellation in that sum because each power $Q^r$ provides only degree $r$ homogeneous monomials with positive coefficients as polynomials in the generators $v_j$. Therefore the proof reduces to check that $Nkq+p$ and $Nkq-p$ belong to the additive semigroup generated by the gaps $N_i-N$ for $k$ large enough.  Since the normalized integers $(N_i-N)/p$ are coprime,  every large enough multiple of $p$, and in particular $Nkq+p$ and $Nkq-p$ for $k$ large enough, belong to the additive semigroup generated by the gaps $N_i-N$, which ends the proof.

\end{proof}
\begin{remark}
 We were unable to prove Lemma $14$ after the further specialization $u_i=1$.
\end{remark}
\begin{lemma}
       Suppose that $ X= \{ n \in \mathbb{N} | supp([L^{-n}_+,L]) \subset supp(L) \} $ is infinite and that $[L^{-n}_+,L]$ with $n \in X$ spans an infinite-dimensional subspace of  $\mathcal{A}[\mathcal{S}, \mathcal{S}^{-1}]$. Then $supp(L)$ is arithmetic.
\end{lemma}
\begin{proof}
    If $k$ is not divisible by $q$, then the support of $[L^{-k}_+,L]$ is orthogonal to the support of $L$. Hence we must have infinitely many integers in $X$ that are divisible by $q$.
    For such an integer, the leading term of $[(L^{-qk})_{-},L]$ is $(1-\mathcal{S}^{-N})(G(k)) \mathcal{S}^{-N-p}$. This is nonzero for $k$ large enough by Lemma $14$. Since $qk \in X$ and 
    \[
    [(L^{-qk})_{-},L]=-[(L^{-qk})_{+},L],
    \]
    we must have $-N-p$ in the support of $L$. In other words, 
    \[ N_1=N+p.\]
    Now, suppose that for some $1 \leq i <m$, we have 
    \[ N_j=N+jp \text{  for all  } j \leq i \, , \, \, N_{i+1} > N+(i+1)p.\]
    Then, consider the operator
    \[ \bar{L}=\mathcal{S}^{-N}+u_1\mathcal{S}^{-N_1}+...+u_{i+1}\mathcal{S}^{-N_{i+1}}.\]
    Applying Lemma $14$ to $\bar{L}$, we obtain $\bar{F}(k) \neq 0$ for sufficiently large $k$. Therefore, for large $k$, the coefficient of $\mathcal{S}^{p-N_{i+1}}$ is nonzero in 
    \[ [({\bar{L}}^{-qk})_{-},\bar{L}]=-[({\bar{L}}^{-qk})_{+},\bar{L}].\]
    This implies that for large $k$ the coefficient of $\mathcal{S}^{p-N_{i+1}}$ is nonzero in 
    \[ [(L^{-qk})_{-},L].\]
    Hence $p-N_{i+1}$ must be in the support of $L$ which is a contradiction. We have shown that 
    \[ supp(L)=\{-N,-N-p,...,-N-mp\}.\]
\end{proof}
\begin{proof} (\emph{of Theorem $10$})
    By Lemmas $13$ and $15$, the support of $L$ must be arithmetic. Conversely, the evolutionary derivations \eqref{eqeqeqeq} are well-defined and commute by a standard argument. For each pair of integers $k,k' \geq 1$, $\partial_k(h_{k'})$ lies in the image of $\mathcal{S}-1$ by Lemma $3$ since it is the residue of a commutator.
\end{proof}

\begin{example}
The Beffa-Wang lattice is associated with the lattice $W_m$-algebra for $m >2$ (\cite{BW13}, \cite{HI97}). Its first flow is given by 
\[
\begin{cases}
v_{1,t}={v_1}\left( \mathcal{S}^{-1}-\mathcal{S}^{m-1}\right)(v_2),\\[0.4em]
v_{i,t}=\left(1-\mathcal{S}^{-1}\right) (v_{i+1})
+v_i\left( \mathcal{S}^{-1}-\mathcal{S}^{i-1}\right) (v_2),
\qquad i=2,3,\ldots,m-1
\end{cases}
\]
where we use the convention $v_m=v_1$ following the notations of \cite{DSKW18}. For $m=3$, this reduces to the Belov-Chaltikian lattice \cite{BC93}. Theorem $10$ shows that the arithmetic support $(-1,1,m-1)$ naturally gives rise to a hierarchy based upon the scalar difference operator 
\begin{equation}\label{eqeq}
L=\mathcal{S}^{-1}+v_2 \mathcal{S}^{-2}+\mathcal{S}^{-1}(v_3) \mathcal{S}^{-3}+...+\mathcal{S}^{3-m}(v_{m-1}) \mathcal{S}^{1-m}+\mathcal{S}^{-m-2}(v_{1}) \mathcal{S}^{-m}.
\end{equation} 
A direct computation gives
\[ (L^{-1})_+=\mathcal{S}-\mathcal{S}(v_2).\]
and shows that the corresponding Lax equation
\begin{equation}\label{e}
L_t=[(L^{-1})_+,L]
\end{equation}
coincides with the Beffa-Wang equation. For example, taking the $\mathcal{S}^{-i}$ component of equation \eqref{e} for a generic $i$ reads 
\[ \mathcal{S}^{2-i}(v_{i,t})=(\mathcal{S}^{2-i}-\mathcal{S}^{1-i})(v_{i+1})+\mathcal{S}^{2-i}(v_i)(\mathcal{S}^{1-i}-\mathcal{S})(v_2).\]
Therefore the Beffa-Wang equation extends naturally to the infinite hierarchy of commuting flows 
\[\partial_k(L)=[(L^{-k})_+,L], \, \, k \geq 1\]
provided by Theorem $10$. Moreover, it possesses infinitely many conserved densities
\[h_k= \text{ Res } L^{-k}, \, \, k \geq 1.\]
Theorem $10$ shows that the Beffa-Wang hierarchy naturally belongs to the class of arithmetic-support hierarchies. In particular, the Beffa-Wang hierarchy admits the scalar difference Lax representation \eqref{eqeq}, which does not appear to have been previously observed in the literature. 
\end{example}

\section{Conclusion}
In this paper, we proved that monic finite-band scalar difference operators with independent coefficients admit infinitely many support-preserving Lax flows if and only if their support is arithmetic. Consequently, arithmetic supports provide a complete classification of finite-band scalar Lax hierarchies in the generic setting considered here. Every such support gives rise to an infinite hierarchy of local commuting flows and infinitely many conserved densities.

Finite-band scalar Lax hierarchies with independent coefficients are therefore parametrized by three integers $(N,p,m)$, corresponding respectively to the leading order, the common size of the gaps in the support, and the number of generators. This places a variety of classical integrable systems, including the Volterra, Toda and Narita--Itoh--Bogoyavlensky hierarchies, within a common arithmetic framework \cite{K85,Y06}. It also complements the modern theory of integrable difference operators developed in \cite{AP11}. At the same time, the theorem produces infinitely many additional integrable hierarchies associated with arithmetic supports that do not appear to have been previously studied.

Several natural questions remain open. First, it would be interesting to determine whether analogous rigidity phenomena persist for matrix-valued or rational difference operators. In view of the rich theory of rational difference operators and their integrable hierarchies \cite{AP11}, this appears to be a particularly natural direction for further investigation. Second, the Hamiltonian structures associated with the hierarchies constructed here remain to be understood. Finally, our proof of the nonvanishing result relies on the assumption that the coefficients are algebraically independent. We expect that an analogous statement remains valid after specializing all generators to $1$, but we do not currently know a proof.

\end{document}